\documentclass[letterpaper,10pt,conference]{ieeeconf}
\IEEEoverridecommandlockouts
\usepackage{amsmath,amssymb,amsfonts}
\usepackage{cite}
\usepackage{algorithmic}
\usepackage{graphicx}
\usepackage{textcomp}
\usepackage{xcolor}
\usepackage{caption}
\usepackage{verbatim}

\newtheorem{theorem}{Theorem}

\newcommand{\bmtx}{\begin{bmatrix}}
\newcommand{\emtx}{\end{bmatrix}}
\newcommand{\bsmtx}{\left[ \begin{smallmatrix}} 
\newcommand{\esmtx}{\end{smallmatrix} \right]}
\newcommand{\field}[1]{\mathbb{#1}}
\newcommand{\R}{\field{R}}
\newcommand{\N}{\field{N}}
\newcommand{\pjs}[1]{{\color{red}{(PJS: #1)}}}

\renewcommand{\R}{\mathbb{R}}
\newcommand{\mcl}[1]{\mathcal{#1}}
\newcommand{\mbf}[1]{\mathbf{#1}}

\def\BibTeX{{\rm B\kern-.05em{\sc i\kern-.025em b}\kern-.08em
    T\kern-.1667em\lower.7ex\hbox{E}\kern-.125emX}}

\begin{document}
\title{\LARGE \bf Time-Transformation-Based Analysis of Systems \\ with Periodic Delay via Perturbative Expansion}

\author{Jungbae Chun$^{1}$, Sengiyumva Kisole$^{2}$, Matthew M. Peet$^{2}$, and Peter Seiler$^{1}$
	\thanks{$^{1}$J. Chun and P. Seiler are with the Department of Electrical Engineering \& Computer Science at the University of Michigan ({\tt\small jungbaec@umich.edu} and
		{\tt\small pseiler@umich.edu}). $^{2}$S. Kisole and M. M. Peet are with the School for the Engineering of Matter, Transport and Energy at Arizona State University ({\tt \small sengi.kisole@asu.edu} and {\tt\small mpeet@asu.edu}). 
  }
}

\maketitle

\begin{abstract}
It is difficult to analyze the stability of systems with time-varying delays.  One approach is to construct a time-transformation that converts the system into a form with a constant delay but  with a time-varying scalar appearing in the system matrices.  The stability of this transformed system can then be analyzed using methods to bound the effect of the time-varying scalar.  One issue is that this transformation is non-unique and requires the solution of an Abel equation.  A specific time-transformation typically must be computed numerically. We address this issue by computing an explicit, although approximate, time-transformation for systems where the delay has a constant plus small periodic term.  We use a perturbative expansion to construct our explicit solutions.  We provide a simple numerical example to illustrate the approach. We also demonstrate the use of this time-transformation to analyze stability of the system with this class of periodic delays.

\end{abstract}


\section{Introduction}

Systems with time-varying delays are challenging to analyze. 
This includes issues associated with causality (when the delay
is fast varying), existence and uniqueness, and minimality  of the state representation. Moreover, systems with time-varying delays,
unlike systems with constant delays, lack a fixed notion of state space\cite{verriest2010well, verriest2011inconsistencies, verriest2012SS1}. Furthermore, there are fewer tools to analyze stability and performance of systems with time-varying delays.
This is in contrast to systems with constant delays for which there are well established tools, including Integral Quadratic Constraints (IQCs), Linear Matrix Inequalities (LMIs), and Partial Integral Equations (PIEs).

These issues motivate the use of a time-transformation to convert a system with a time-varying delay into an equivalent (transformed) system with a constant delay \cite{heard1975change, neuman1990transformation, vcermak1995continuous}. A transformation $h$ can be constructed, under some technical assumptions, to map the original time to a new time variable. The new transformed system has a constant delay but a time-varying scalar factor appears in the state matrices.

This construction of an appropriate time-transformation requires solving an Abel equation. Existing methods rely on two steps~\cite{nah2020normalization, brunner2009time1}. 
First, a seed function is specified on an initial interval. Next, the time transformation is propagated forward, recursively, from this seed function.  We review this construction in Section~\ref{sc:probform}. This approach often requires numerical methods and an explicit solution for a time transformation can be found only for special classes of delays. Also note that time transformations are not unique, i.e., different seed functions lead to different valid time-transformations.  A complete parameterization of all possible seed functions is given in \cite{Kisole2025parameterization}.  

The main contribution of this paper is to provide an explicit, although approximate, time-transformation for a certain class of delays. Specifically, we consider systems where the delay consists of a constant plus a periodic term with a small amplitude $\epsilon$. Such explicit approximate solutions are valuable for certain applications, e.g.,  the suppression of regenerative chatter in turning and milling processes. In these systems, a deliberate periodic variation of the spindle speed around its nominal value introduces a small periodic component into the regenerative delay. This modulation can significantly enlarge the region of stable operation and, in some cases, stabilize systems that would otherwise be unstable under a constant delay. Theoretical and experimental studies have confirmed this stabilizing effect \cite{michiels2005stabilization, insperger2004stability, seguy2010stability, zatarain2008stability}. We use a perturbative expansion in $\epsilon$ to construct first- and second-order approximations for the time-transformation. These main results are presented in Section~\ref{sc:PE} with a simple example given in Section~\ref{sc:PEsindelay}. We also illustrate the use of these approximate time-transformations for stability analysis
(Sections~\ref{sc:pipeline} and \ref{sec:stabnumex}).  This analysis uses the Partial Integral Equation (PIE) representation \cite{peet2021representation} and PIETOOLS \cite{shivakumar2025pietools}. Conclusions and future work are given in Section~\ref{sc:concfw}.

\subsection*{Notation}
$\R$, $\R_{+}$, and $\N$ denote
the sets of real, (strictly) positive real, and natural numbers. $C(A,B)$ denotes the set of functions $f:A \to B$ that are continuous and $C^k(A,B)$ denotes the set of $k$-times continuously differentiable functions for $k \in \N$. For $[a, b] \subseteq \R$, $L_2^n[a, b]$ denotes the space of functions $f:[a, b] \to \R^n$ satisfying $\|f\| < \infty$ where 
$\|f\|:= \left[\int_a^b f(s)^{\top}f(s)ds \right]^{0.5}$.
The composition of two functions $f$ and $g$ is defined as $(f \circ g)(x) = f(g(x))$. For $k \in \N$, $f^k$ denotes the $k$-fold composition of $f$. For $p, q \in \R$, we denote $p \ll q$ to indicate that $p$ is much less than $q$, i.e., $p/q$ is assumed to be sufficiently small. We denote $f_1(\epsilon) = O(f_2(\epsilon))$, if there exists $M \in \R_+$ such that $|f_1(\epsilon)| \leq M |f_2(\epsilon)|$ for $\epsilon$ sufficiently small.


\section{Problem Formulation and Preliminaries}
\label{sc:probform}
Consider a linear Delay Differential Equation (DDE):
\begin{align}
\label{eq:DDEwithTVdelay}
    \begin{split}
	&\dot{x}(t) = A_0 x(t) + A_1 x(t-\tau(t)), \, \forall t \geq 0, \\
	&x(t) = x_0(t), \, \forall t \in [-\tau(0), 0],
    \end{split}
\end{align}
where $x(t)\in \R^{n_x}$ is the state vector at time $t$, $A_0 \in \R^{n_x \times n_x}$ and $A_1 \in \R^{n_x \times n_x}$ are state matrices, and $x_0 \in C([-\tau(0),0], \R^{n_x})$ is the initial function. The delay $\tau \in C^1(\R, \R_{+})$ is assumed to be bounded with $\dot{\tau}(t) < 1$ for all $t \in \R$. Equation \eqref{eq:DDEwithTVdelay} has a time-varying delay but constant state matrices. We can construct a time-transformation to convert this to a DDE with a constant delay, but with a time-varying factor in the state matrices. This is stated precisely in the next theorem. This result is an extended version of Theorem 1 in \cite{nah2020normalization}, which covers one direction of the equivalence. The proof for both direction is given in Lemma 1 of \cite{Kisole2025parameterization}.
\begin{theorem}
\label{Theorem:time-transformation}
Let $\tau \in C^1(\R, \R_{+})$ be a bounded function with $\dot{\tau}(t) < 1$ for all $t \in \R$. For a given $\tau^* \in \R_{+}$, there exists a strictly increasing, invertible function $h \in C^1([-\tau^*, \infty), \R)$ which satisfies 
\begin{align}
\label{eq:hcond1}
	&h(\lambda) - \tau(h(\lambda)) = h(\lambda - \tau^*), \, \forall \lambda \geq 0,\\
\label{eq:hcond2}
	&h(0) = 0.
\end{align}
Moreover, $x$ satisfies \eqref{eq:DDEwithTVdelay} for a given initial function $x_0 \in C([-\tau(0),0], \R^{n_x})$ if and only if $\bar{x} = x \circ h$ satisfies
\begin{align}
\label{eq:transformed DDE}
    \begin{split}
        &\dot{\bar{x}}(\lambda) = \dot{h}(\lambda) A_0 \bar{x}(\lambda) + \dot{h}(\lambda) A_1 \bar{x}(\lambda - \tau^*), \, \forall \lambda \geq 0,\\
	&\bar{x}(\lambda) = \bar{x}_0(\lambda), \, \forall \lambda \in [-\tau^*, 0], 
    \end{split}
\end{align}
where $\bar{x}_0 = x_0 \circ h \in C([-\tau^*,0], \R^{n_x})$ and $\lambda = h^{-1}(t)$.
\end{theorem}



Equation \eqref{eq:hcond1} is an Abel equation. The solution $h$ is an invertible function that transforms between the original time $t$ and the new time measure $\lambda = h^{-1}(t)$. Note that $\dot{\tau}(t) < 1$ together with $\tau$ bounded implies that the function $g(t) = t - \tau(t)$ is strictly increasing and $g(\R) = \R$. Hence, the inverse function $g^{-1}:\R \rightarrow \R$ is well-defined and strictly increasing. It was shown in \cite{nah2020normalization} that such a time-transformation $h: [-\tau^*, \infty) \rightarrow \R$ can be constructed in two steps. First, a function $\phi \in C^1([-\tau^*, 0], \mathbb{R})$ is constructed to satisfy:
\begin{align}
\label{eq:condition on phi}
\begin{split}
    & \phi(0) = 0, \\
    & \phi(-\tau^*) = -\tau(0), \\
    & \dot{\phi}(0) = \frac{\dot{\phi}(-\tau^*)}{1 - \dot{\tau}(0)}, \\
    & \dot{\phi}(\lambda) > 0, \, \forall \lambda \in [-\tau^*, 0].
\end{split}
\end{align}
Next, the time-transformation $h: [-\tau^*, \infty) \rightarrow \R$ is constructed by propagating forward from the function $\phi$:
\begin{align}
\label{eq:construct}
h(\lambda) =
    \begin{cases}
	\phi(\lambda) & \mbox{ for } \lambda \in [-\tau^*, 0), \\
	g^{-1}(h(\lambda-\tau^*)) & \mbox{ for } \lambda \geq 0.
    \end{cases}
\end{align}
From \eqref{eq:construct}, it follows that $\forall k \in \N$,
\begin{align}
\label{eq:numerical h constructed by induction}
h(\lambda) = (g^{-1})^k(\phi(\lambda - k\tau^*)) \mbox{ for } \lambda \in [(k-1)\tau^*, k\tau^*).
\end{align}
Theorem~\ref{Theorem:time-transformation} states that such a transformation exists, but  the conditions \eqref{eq:condition on phi} do not uniquely specify $\phi$, referred to as the seed function. Hence the time-transformation is not unique. A complete parameterization of seed functions is given in \cite{Kisole2025parameterization}.
In our paper, we provide one explicit, although approximate, time-transformation for a class of periodic delays. 

\section{Main result: Perturbative expansions for $h$}
\label{sc:PE}
A closed-form analytic expression for $h$ cannot be obtained from \eqref{eq:numerical h constructed by induction} unless the exact closed-form expression of $g^{-1}$ corresponding to the given $\tau(t)$ is known. As a consequence, the factor $\dot{h}(\lambda)$ in the transformed system \eqref{eq:transformed DDE} cannot be computed analytically. The remainder of this section presents one explicit, although approximate, choice for $h$. This expression is obtained for systems where the delay is a constant plus small periodic term of the form:
\begin{align}
\label{eq:tau with perturbation}
\tau(t) = \tau_0 + \epsilon \tilde{\tau}(t).
\end{align}
We assume $\tilde{\tau} \in C^1(\R,[-1,1])$ is a periodic function with fundamental period $P \in \R_{+}$. Its Fourier series is given by:
\begin{align}
\label{eq:tauFS}
\tilde{\tau}(t) = \sum_{k = -\infty}^{\infty} a_k e^{jk\omega t},
\end{align}
where $\omega = \frac{2\pi}{P}$ and $a_k = \frac{1}{P} \int_0^P \tilde{\tau}(t)e^{-jk\omega t}dt$. The approximate time-transformation is obtained via perturbative analysis with both first- and second-order expansions in $\epsilon$.

\subsection{First-Order Expansion}
\label{PE:first-order}
Our first main result is an approximate solution for the time-transformation using a first-order expansion in $\epsilon$.
\begin{theorem}
\label{Theorem:PE until first-order}
Let $\tau \in C^1(\R, \R_{+})$ be a bounded function in the form of \eqref{eq:tau with perturbation}, where $\tilde{\tau}$ is a periodic function with fundamental period $P \in \R_{+}$, and its Fourier series is given in \eqref{eq:tauFS}. We assume: (i) $\dot{\tau}(t) < 1$ for all $t \in \R$, (ii) $0 < \epsilon \ll \tau_0$, (iii) $a_0 = 0$, and (iv) $\omega\tau_0 \neq m\pi$ for all $m \in \N$. For a given $\tau^* \in \R_{+}$, one approximate time-transformation $h: [-\tau^*, \infty) \rightarrow \R$ that satisfies \eqref{eq:hcond1}-\eqref{eq:hcond2} in Theorem~\ref{Theorem:time-transformation} is given by:
\begin{align}
\label{eq:hFOExpansion}
\begin{split}
    h(\lambda) &= \left( \frac{\tau_0}{\tau^*} \right) \lambda + \epsilon \sum_{k = -\infty}^{\infty} b_k e^{jk\omega\left( \frac{\tau_0}{\tau^*} \right)\lambda}
    + O(\epsilon^2),
\end{split}
\end{align}
where
\begin{align}
\label{eq:bk}
&b_k = \frac{a_k}{1 - e^{-jk\omega\tau_0}}, \, \forall k \neq 0, \\
\label{eq:b0}
&b_0 = -\sum_{k \neq 0} \frac{a_k}{1 - e^{-jk\omega\tau_0}}.
\end{align}
\end{theorem}
\begin{proof} 
The time-transformation depends on the perturbation amplitude
$\epsilon$, denoted as $h_\epsilon$.  Use a series expansion
to express this transformation as follows:
\begin{align}
\label{eq:hseries1}
h_\epsilon(\lambda) = h_0(\lambda) + \epsilon h_1(\lambda) + O(\epsilon^2).
\end{align}
The $O$ notation in \eqref{eq:hFOExpansion} and \eqref{eq:hseries1} means that the remainder term can be bounded by $M(\lambda) \epsilon^2$ for some function $M$ and for $\epsilon$ sufficiently small.
The function $h_\epsilon$ must satisfy \eqref{eq:hcond1}-\eqref{eq:hcond2}. 

First, consider the condition in \eqref{eq:hcond2},
$h_\epsilon(0) = 0$. The expansion in
\eqref{eq:hseries1} should satisfy this condition for all sufficiently small $\epsilon>0$. This implies that $h_0(0) = 0$ and $h_1(0) = 0$.

Next, consider the condition in \eqref{eq:hcond1}:
\begin{align}
\label{eq:hepscond1}
h_\epsilon(\lambda) - \tau(h_\epsilon(\lambda)) = h_\epsilon(\lambda - \tau^*).
\end{align}
We can also use a series expansion for the delay term:
\begin{align}
\label{eq:tauseries1}
\begin{split}
    \tau(h_\epsilon(\lambda)) &= \tau_0 + \epsilon \tilde{\tau}(h_\epsilon(\lambda)) \\
    &= \tau_0 + \epsilon \tilde{\tau}(h_0(\lambda)) + O(\epsilon^2).
\end{split}
\end{align}
Substitute the series expansions given in \eqref{eq:hseries1}
and \eqref{eq:tauseries1} back into \eqref{eq:hepscond1}. Collecting all terms on the left side yields:
\begin{align}
\label{eq:hcollect1}
\begin{split}
    &\left(h_0(\lambda) - \tau_0 - h_0(\lambda - \tau^*) \right) \\
    &+ \epsilon \left( h_1(\lambda) - \tilde{\tau}(h_0(\lambda)) - h_1(\lambda - \tau^*) \right) 
    + O(\epsilon^2) = 0.
\end{split}
\end{align}
Again, this condition should be satisfied for all sufficiently small $\epsilon>0$. This implies two separate conditions based on the terms with zeroth and first-order dependence on $\epsilon$:
\begin{align}
\label{eq:h0cond}
& h_0(\lambda) - \tau_0 - h_0(\lambda - \tau^*)  = 0, \\
\label{eq:h1cond}
& h_1(\lambda) - \tilde{\tau}(h_0(\lambda)) - h_1(\lambda - \tau^*) = 0.
\end{align}
One function $h_0: [-\tau^*, \infty) \rightarrow \R$ satisfying $h_0(0)=0$ and \eqref{eq:h0cond} is $h_0(\lambda) = \left( \frac{\tau_0}{\tau^*} \right)\lambda$.
Substituting this into \eqref{eq:h1cond} yields:
\begin{align}
\label{eq:h1condB}
h_1(\lambda) - \tilde{\tau}\left(\left( \frac{\tau_0}{\tau^*} \right)\lambda\right) - h_1(\lambda - \tau^*) = 0.
\end{align}
Next, recall that $\tilde\tau(t)$ has a fundamental period
$P \in \R_{+}$ with Fourier series in \eqref{eq:tauFS}.
Hence $\tilde{\tau}(\left( \frac{\tau_0}{\tau^*} \right)\lambda)$ has a fundamental period, as a function of $\lambda$, given by $P\left( \frac{\tau^*}{\tau_0} \right)$. It is natural to suppose $h_1$ is also a periodic function 
with fundamental period $P\left( \frac{\tau^*}{\tau_0} \right)$. In this case, $h_1$ has a Fourier series representation:
\begin{align}
\label{eq:h1FS}
h_1(\lambda) = \sum_{k = -\infty}^{\infty} b_k e^{jk\omega\left( \frac{\tau_0}{\tau^*} \right)\lambda}.
\end{align}
Substitute the Fourier series for $h_1$ and $\tilde{\tau}$ into \eqref{eq:h1condB}. Comparing the Fourier coefficients yields: 
\begin{align}
\label{eq:FSrelation1}
b_k - a_k - b_k e^{-jk\omega\tau_0} = 0.
\end{align}
This can be solved for each Fourier coefficient $b_k$ with $k \neq 0$:
\begin{align}
\label{eq:bksol}
b_k = \frac{a_k}{1 - e^{-jk\omega\tau_0}}.
\end{align}
Note that this expression is well-posed due to the
assumption $\omega\tau_0 \neq m\pi$ for all $m \in \N$ (equivalently, $e^{-j\omega\tau_0} \neq \pm 1$).

Finally, the condition $h_1(0)=0$ can be used
to solve for the constant $b_0$ in the Fourier series of $h_1$:
\begin{align}
b_0 = -\sum_{k \neq 0} \frac{a_k}{1 - e^{-jk\omega\tau_0}}. 
\end{align}
\end{proof}

\subsection{Second-Order Expansion}
\label{PE:second-order}
Our second main result is an approximate solution for the time-transformation using a second-order expansion in $\epsilon$.
\begin{theorem}
\label{Theorem:PE until second-order}
Under the same assumptions on $\tau(t)$, $\epsilon$, $\tau_0$, $\tilde{\tau}(t)$, and $P$ as in Theorem~\ref{Theorem:PE until first-order}, and for a given $\tau^* \in \R_{+}$, one approximate time-transformation $h: [-\tau^*, \infty) \rightarrow \R$ that satisfies \eqref{eq:hcond1}-\eqref{eq:hcond2} in Theorem~\ref{Theorem:time-transformation} is given by
\begin{align}
\begin{split}
    h(\lambda) &= \left( \frac{\tau_0}{\tau^*} \right)\lambda + \epsilon \sum_{k = -\infty}^{\infty} b_k e^{jk\omega\left( \frac{\tau_0}{\tau^*} \right)\lambda} \\ 
    &\quad + \epsilon^2 \bigg(\left( \frac{m_0}{\tau^*} \right)\lambda + \sum_{k = -\infty}^{\infty} c_k e^{jk\omega\left( \frac{\tau_0}{\tau^*} \right)\lambda}\bigg) \\ 
    &\quad + O(\epsilon^3),
\end{split}
\end{align}
where $b_k$, $b_0$ are given in \eqref{eq:bk}, \eqref{eq:b0} and the remaining coefficients are given by:
\begin{align}
\label{eq:mk}
m_k &= \sum_{l = -\infty}^{\infty} jl\omega a_l b_{k-l}, \, \forall k,\\
\label{eq:ck}
c_k &= \frac{m_k}
           {1 - e^{-jk\omega\tau_0}} , \, \forall k \neq 0, \\
\label{eq:c0}
c_0 &= -\sum_{k \neq 0} \frac{m_k}{1 - e^{-jk\omega\tau_0}}. 
\end{align}
\end{theorem}

\begin{proof}
The proof is similar to that given for Theorem~\ref{Theorem:PE until first-order} except we retain terms up to $\epsilon^2$:
\begin{align}
\label{eq:hseries2}
h_\epsilon(\lambda) = h_0(\lambda) + \epsilon h_1(\lambda) + \epsilon^2 h_2(\lambda) + O(\epsilon^3).
\end{align}
The function $h_\epsilon$ must satisfy \eqref{eq:hcond1}-\eqref{eq:hcond2}. 

First, the expansion in \eqref{eq:hseries2} should satisfy the condition \eqref{eq:hcond2}, $h_\epsilon(0) = 0$, for all sufficiently small $\epsilon>0$.  This implies that $h_0(0) = 0$, $h_1(0) = 0$, and $h_2(0) = 0$.

Next, consider the condition in \eqref{eq:hcond1}:
\begin{align}
\label{eq:hepscond1rep}
h_\epsilon(\lambda) - \tau(h_\epsilon(\lambda)) = h_\epsilon(\lambda - \tau^*).
\end{align}
We can also use a series expansion for the delay term:
\begin{align}
\label{eq:tauseries2}
\tau(h_\epsilon(\lambda)) &= \tau_0 + \epsilon \tilde{\tau}(h_\epsilon(\lambda)) \\
\nonumber
&= \tau_0 + \epsilon \tilde{\tau}(h_0(\lambda)) + \epsilon^2 \tilde{\tau}'(h_0(\lambda)) \, h_1(\lambda) + O(\epsilon^3).
\end{align}
Substitute the series expansions given in \eqref{eq:hseries2}
and \eqref{eq:tauseries2} back into \eqref{eq:hepscond1rep}. Collecting all terms on the left side yields:
\begin{align}
\label{eq:hcollect2}
\begin{split}
    &\left( h_0(\lambda) - \tau_0 - h_0(\lambda - \tau^*) \right) \\
    &+ \epsilon\left( h_1(\lambda) - \tilde{\tau}(h_0(\lambda)) - h_1(\lambda - \tau^*) \right) \\ 
    &+ \epsilon^2\left( h_2(\lambda) - \tilde{\tau}'(h_0(\lambda)) \, h_1(\lambda) - h_2(\lambda - \tau^*) \right) \\
    &+ O(\epsilon^3) = 0.
\end{split}
\end{align}
This should be satisfied for all sufficiently small $\epsilon>0$. Use the same steps as in
the proof of Theorem~\ref{Theorem:PE until first-order} to solve for $h_0$ and $h_1$.  This yields 
$h_0(\lambda) = \left( \frac{\tau_0}{\tau^*} \right)\lambda$ and $h_1(\lambda)$ as in \eqref{eq:h1FS}.


The terms with second-order dependence on $\epsilon$ in 
\eqref{eq:hcollect2} yield one
additional condition:
\begin{align}
\label{eq:h2cond}
h_2(\lambda) - \tilde{\tau}'(h_0(\lambda)) \, h_1(\lambda) - h_2(\lambda - \tau^*) = 0.
\end{align}
To solve this, we need to simplify the second term on the left side. We have shown $h_0(\lambda) = \left( \frac{\tau_0}{\tau^*} \right)\lambda$. Hence, we can use the Fourier series of $\tilde{\tau}$ and $h_1$ (given by \eqref{eq:tauFS} and \eqref{eq:h1FS}, respectively) to show:
\begin{align}
\label{eq:multiplication}
    \tilde{\tau}'\left(\left( \frac{\tau_0}{\tau^*} \right)\lambda\right) h_1(\lambda) 
    = \sum_{k = -\infty}^{\infty} m_k e^{jk\omega\left( \frac{\tau_0}{\tau^*} \right)\lambda},    
\end{align}
where $m_k$ is defined in \eqref{eq:mk}. This term corresponds to a Fourier series with fundamental period, as a function of $\lambda$, given by $P\left( \frac{\tau^*}{\tau_0} \right)$. 
One important point is that $m_0\ne 0$ in general in this Fourier Series.   To account for this term,  define $\hat{h}_2(\lambda) := h_2(\lambda) - \left( \frac{m_0}{\tau^*} \right)\lambda$. Then, \eqref{eq:h2cond} is equivalent to 
\begin{align}
\label{eq:hath2cond}
\hat{h}_2(\lambda) + \left(m_0 - \sum_{k = -\infty}^{\infty} m_k e^{jk\omega\left( \frac{\tau_0}{\tau^*} \right)\lambda} \right)
= \hat{h}_2(\lambda - \tau^*).
\end{align}
It is again natural to suppose that $\hat{h}_2$ is also a periodic function with the same fundamental period.  In this case, $\hat{h}_2$ admits the Fourier series representation:
\begin{align}
\label{eq:hath2FS}
\hat{h}_2(\lambda) = \sum_{k = -\infty}^{\infty} c_k e^{jk\omega\left( \frac{\tau_0}{\tau^*} \right)\lambda}.
\end{align}
Substitute the Fourier series for $\hat{h}_2$ and the Fourier Series \eqref{eq:multiplication} into \eqref{eq:hath2cond}. Comparing the Fourier coefficients yields:
\begin{align}
\label{eq:FSrelation2}
\begin{split}
    &k = 0: c_0 + (m_0 - m_0) = c_0, \\
    &k \neq 0: c_k - m_k = c_k e^{-jk\omega\tau_0}.
\end{split}
\end{align}
The condition for $k=0$ is satisfied for any choice of $c_0$.
Solving the conditions for $k\ne 0$ gives:
\begin{align}
\label{eq:cksol}
c_k = \frac{m_k}{1 - e^{-jk\omega\tau_0}}.
\end{align}
Note that this expression is well-posed due to the
assumption $\omega\tau_0 \neq m\pi$ for all $m \in \N$ (equivalently, $e^{-j\omega\tau_0} \neq \pm 1$).

Finally, the condition $\hat{h}_2(0) = h_2(0) = 0$ can be used to solve for the remaining constant $c_0$. This yields:
\begin{align}
c_0 = -\sum_{k \neq 0} \frac{m_k}{1 - e^{-jk\omega\tau_0}}. 
\end{align}
Thus, the second-order term is $h_2(\lambda)=\left(\frac{m_0}{\tau^*}\right) \lambda + \hat{h}_2(\lambda)$ with $\hat{h}_2$
given by the Fourier series \eqref{eq:hath2FS}.
\end{proof}

\section{Pipeline for the Stability Analysis}
\label{sc:pipeline}
Consider the DDE \eqref{eq:DDEwithTVdelay} with delay $\tau(t)$ satisfying the conditions in Theorem~\ref{Theorem:time-transformation}. Apply a strictly increasing transformation $h$ to give a DDE \eqref{eq:transformed DDE} with a constant delay $\tau^*$. Assume there exists constants $0<h_l \le h_u < \infty$ such that:
\begin{align*}
h_l \leq \dot{h}(\lambda) \leq h_u, \, \forall \lambda \in [-\tau^*,\infty).
\end{align*}
Define the average and amplitude of these bounds by $\bar{h} := \frac{1}{2}(h_u + h_l)$ and $\gamma := \frac{1}{2}(h_u - h_l)$, respectively. The variation of $\dot{h}$ about the average bound is $\delta(\lambda) := \dot{h}(\lambda) - \bar{h}$. We can rewrite the transformed system \eqref{eq:transformed DDE} in terms of $\delta$ as follows:
\begin{align}
\begin{split}
\label{eq:transformed DDE rewrite}
    \dot{\bar{x}}(\lambda) &= \bar{h}\, A_0 \bar{x}(\lambda) + \bar{h} \, A_1 \bar{x}(\lambda - \tau^*) \\
    &\quad + \delta(\lambda) \, (A_0 \bar{x}(\lambda) + A_1 \bar{x}(\lambda - \tau^*)), \, \forall \lambda \geq 0, \\
    \bar{x}(\lambda) &= \bar{x}_0(\lambda), \, \forall \lambda \in [-\tau^*, 0].
\end{split}
\end{align}
We can recast this into a feedback interconnection $[G,\Delta$]:
\begin{align}
\label{eq:feedback}
\begin{split}
    &G \begin{cases}
        \dot{\bar{x}}(\lambda) = \bar{h} A_0 \bar{x}(\lambda) + \bar{h} A_1 \bar{x}(\lambda - \tau^*) 
        + v(\lambda), \\ 
        \bar{x}(\lambda) = \bar{x}_0(\lambda), \, \forall \lambda \in [-\tau^*, 0], \\
        w(\lambda) = A_0 \bar{x}(\lambda) + A_1 \bar{x}(\lambda - \tau^*),
    \end{cases}\\
        & \quad v(\lambda) = \Delta(w(\lambda)) := \delta(\lambda)\, w(\lambda).
\end{split}
\end{align}
The uncertainty is bounded by the amplitude of the bounds:
\begin{align*}
\|\delta\|_{\infty} \leq \underset{\lambda \geq -\tau^*}{\sup} |\dot{h}(\lambda) - \bar{h}| \leq \gamma.
\end{align*}

By Theorem~\ref{Theorem:time-transformation}, the stability of the feedback interconnection $[G, \Delta]$ is equivalent to the stability of the original DDE \eqref{eq:DDEwithTVdelay}. Consequently, \eqref{eq:feedback} enables the analysis of \eqref{eq:DDEwithTVdelay} using classical robust stability tools, such as IQCs \cite{megretski2002system}. Finite-dimensional approaches, such as discussed in \cite{pfifer2015integral}, are also applicable. However, in this section, we review a stability analysis method using the PIE representation \cite{peet2021representation}.

As shown in \cite{peet2021representation}, an equivalent ODE-PDE system of nominal system $G$ is given by:
\begin{align}
\nonumber
	&\dot{\bar{x}}(\lambda) = \bar{h}(A_0+A_1)\bar{x}(\lambda) - \bar{h} A_1 \int_{-1}^0 \partial_s \phi(\lambda, s)ds + v(\lambda), \\
\label{eq:ODE-PDE}
    &\partial_\lambda \phi (\lambda, s) = \frac{1}{\tau^*} \partial_s \phi (\lambda,s), \\
\nonumber
	&w(\lambda) = (A_0 + A_1 ) \bar{x}(\lambda) - A_1 \int_{-1}^0 \partial_s \phi(\lambda, s)ds, \\
\nonumber
    &\phi(\lambda, 0) = \bar{x}(\lambda), \, \bar{x}(0) = \bar{x}_0(0), \, \phi(0, s) = \bar{x}_0(\tau^*s),
\end{align}
where $\lambda \geq 0$ and $s \in [-1, 0]$.

We say $\mcl P$ is a Partial Integral (PI) operator on the Hilbert Space $\mbf Z^{n, m} := \R^n \times L_2^m[-1, 0]$ for given matrix $P$, bounded functions $Q_1, Q_2, R_0$, and separable functions $R_1, R_2$ such that 
\begin{align*}
\Biggl( \mcl P \bmtx P & Q_1 \\ Q_2 & \{R_i\}  \emtx \bmtx x \\ \mbf \Phi \emtx \Biggr)(s) = \bmtx Px + \int_{-1}^0 Q_1(s) \mbf \Phi(s) ds \\ Q_2(s)x + (\mathcal{P}_{\{R_i\}} \mathbf{\Phi})(s) \emtx,
\end{align*}
where
\begin{align*}
&(\mcl P_{\{R_i\}} \mbf \Phi)(s) := \\ 
&R_0(s) \mbf \Phi(s) + \int_{-1}^s R_1(s, \theta) \mbf \Phi(\theta) d\theta + \int_s^0 R_2(s, \theta) \mbf \Phi(\theta) d\theta.
\end{align*}
As discussed in \cite{peet2021representation}, the ODE-PDE system \eqref{eq:ODE-PDE} can be represented as PIE using PI operators as
\begin{align}
\label{eq:PIE}
\begin{split}
    \mcl T \mbf{\dot{\bar{x}}}(\lambda) &= \mcl A \mbf{\bar{x}}(\lambda) + \mcl B v(\lambda), \\
    w(\lambda) &= \mcl C \mbf{\bar{x}}(\lambda) + \mcl D v(\lambda),
\end{split}
\end{align}
where
\begin{align*}
&\bar{\mbf{x}}(\lambda) = \bmtx \bar{x}(\lambda) \\ \partial_s \phi(\lambda,\cdot) \emtx, \qquad \quad \; \; \mcl T = \mcl P \bmtx I & 0 \\ I & \{0, 0, -I\} \emtx, \\
&\mcl A = \mcl P \bmtx \bar{h} (A_0+ A_1) & -\bar{h} A_1 \\ 0 & \{\frac{1}{\tau^*}, 0, 0\}\emtx, \; \mcl B = \mcl P \bmtx I & \emptyset \\ 0 & \{\emptyset\} \emtx, \\
&\mcl C = \mcl P \bmtx A_0+ A_1 & -A_1 \\ \emptyset &  \{\emptyset\}\emtx, \qquad \quad \; \; \: \mcl D = \mcl P \bmtx 0 & \emptyset \\ \emptyset & \{\emptyset\} \emtx.
\end{align*}
Finally, the stability analysis can be performed by using Corollary 8 (Augmented KYP lemma) in \cite{talitckii2023integral} and PIETOOLS \cite{shivakumar2025pietools}.  We use the hard IQC defined by $\mcl K = \bmtx \gamma^2 \mcl P & \mcl R \\ \mcl R^* & -\mcl P \emtx$ and $\Psi = I$ since $\delta$ is a varying parameter with $\|\delta\|_{\infty} \leq \gamma$.

\section{Illustrative Examples}
\label{sc:example}

\subsection{Example With Sinusoidal Delays}
\label{sc:PEsindelay}

Consider a sinusoidal delay  $\tau \in C^1(\mathbb{R}, \mathbb{R}_{+})$ given by:
\begin{align}
\label{eq:tauex}
\tau(t) = \tau_0 + \epsilon\sin(\omega t).
\end{align}
We assume $\epsilon\omega < 1$, which ensures that $\dot{\tau}(t) < 1$ for all $t$. In addition, we assume  $0 < \epsilon \ll \tau_0$ 
and $\omega\tau_0 \neq m\pi$ for all $m \in \N$.  This ensures that
we can apply the perturbation theorems in Section~\ref{sc:PE}. The fundamental period is $P= \frac{2\pi}{\omega}$ and the Fourier series of the periodic term is
\begin{align*}
\tilde{\tau}(t) := \sin(\omega t) = \frac{1}{2j}e^{j\omega t} - \frac{1}{2j}e^{-j\omega t}.
\end{align*}
Thus, the Fourier coefficients of $\tilde{\tau}(t)$ are given by $a_1=\frac{1}{2j}$, $a_{-1}=-\frac{1}{2j}$, and $a_k=0$ for $k\ne \pm 1$.

Select $\tau^* = \tau_0$ for the constant delay in the transformed system.  This choice is the mean value of $\tau(t)$ and it simplifies the expressions in the first and second order expansions.

Theorem~\ref{Theorem:PE until first-order} gives the first-order expansion.  The coefficients in the expansion are given by:
\begin{align}
b_k = \begin{cases}
        \frac{1}{2j(1-e^{-j\omega\tau_0})} & \mbox{ for } k = 1, \\
        -\frac{1}{2j(1-e^{j\omega\tau_0})} & \mbox{ for } k = -1, \\
        -(b_{-1}+b_1) = \frac{\cot(\frac{\omega\tau_0}{2})}{2} & \mbox{ for } k = 0, \\
        0 & \mbox{ for } k \neq 0, \, \pm 1.
    \end{cases}
\end{align}
The first-order expansion, with these coefficients, can be simplified to the following expression:
\begin{align*}
    h(\lambda) &= \lambda + \epsilon \left(\frac{\cos(\frac{\omega\tau_0}{2}) - \cos(\omega\lambda + \frac{\omega\tau_0}{2})}{2\sin(\frac{\omega\tau_0}{2})}\right) + O(\epsilon^2). 
\end{align*}
The corresponding derivative is given by:
\begin{align*}
\dot{h}(\lambda) \approx 1 + \epsilon \left(\frac{\omega\sin(\omega\lambda + \frac{\omega\tau_0}{2})}{2\sin(\frac{\omega\tau_0}{2})}\right).
\end{align*}

Similarly, Theorem~\ref{Theorem:PE until second-order} can be used to  obtain a second-order expansion for the time-transformation:
\begin{align*}
\begin{split}
    h(\lambda) &= \lambda + \epsilon \left(\frac{\cos(\frac{\omega\tau_0}{2}) - \cos(\omega\lambda + \frac{\omega\tau_0}{2})}{2\sin(\frac{\omega\tau_0}{2})}\right) \\
    &\quad + \epsilon^2 \bigg(-\frac{\omega\cot(\frac{\omega\tau_0}{2})}{4\tau_0}\lambda \\
    &\qquad \qquad + \frac{\omega\left(\sin(\frac{3\omega\tau_0}{2}) - \sin(2\omega\lambda + \frac{3}{2}\omega\tau_0)\right)}{8\sin(\frac{\omega\tau_0}{2})\sin(\omega\tau_0)} \\ 
    &\qquad \qquad + \frac{\omega\cot(\frac{\omega\tau_0}{2})\left(\sin(\omega\lambda + \frac{\omega\tau_0}{2}) - \sin(\frac{\omega\tau_0}{2})\right)}{4\sin(\frac{\omega\tau_0}{2})} \bigg) \\
    &\quad + O(\epsilon^3).
\end{split}
\end{align*}
The corresponding derivative is given by:
\begin{align*}
\dot{h}(\lambda) &\approx 1 + \epsilon \left(\frac{\omega\sin(\omega\lambda + \frac{\omega\tau_0}{2})}{2\sin(\frac{\omega\tau_0}{2})}\right) \\ &\quad + \epsilon ^2 \bigg(-\frac{\omega\cot(\frac{\omega\tau_0}{2})}{4\tau_0} - \frac{\omega^2\cos(2\omega\lambda + \frac{3}{2}\omega\tau_0)}{4\sin(\frac{\omega\tau_0}{2})\sin(\omega\tau_0)} \\
&\qquad \qquad + \frac{\omega^2\cot(\frac{\omega\tau_0}{2})\cos(\omega\lambda + \frac{\omega\tau_0}{2})}{4\sin(\frac{\omega\tau_0}{2})} \bigg).
\end{align*}

Figures~\ref{fig:hdotApproxSmallEps} and \ref{fig:hdotApproxLargeEps}
show plots of $\dot{h}(\lambda)$ for the first- and second-order approximations, as well as the actual time transformation computed numerically with different parameter values.
Both approximations have a fundamental period of $P= \frac{2\pi}{\omega}$. The actual time transformation is numerically computed via \eqref{eq:construct} after determining the seed function that satisfies the conditions in \eqref{eq:condition on phi}. The seed function fits the second-order approximation over the interval $[-\tau^*, 0]$ and has the following form: 
\begin{align*}
\phi(\lambda) = p\lambda + q_0 + \sum_{k = 1}^3 \bigg(q_k\cos(k\omega\lambda) + r_k\sin(k\omega\lambda)\bigg).
\end{align*}
Figure~\ref{fig:hdotApproxSmallEps}  uses a small amplitude $\epsilon=0.01$ for the periodic term.  In this case, the first- and second-order approximations, as well as the actual time transformation are nearly identical. Figure~\ref{fig:hdotApproxLargeEps} uses a larger amplitude $\epsilon=0.1$ leading to more noticeable differences between the two approximations and the actual time transformation.

\begin{figure}[htbp]
\centering
\includegraphics[width = 0.47\textwidth]{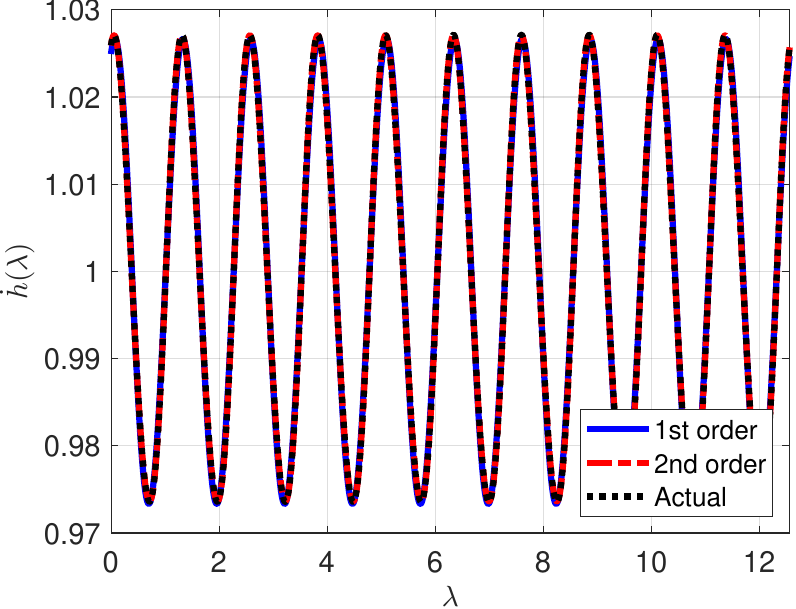}
\caption{First-order (solid blue) and second-order (dashed red) approximations of $\dot{h}(\lambda)$ with $\tau^* = \tau_0 = 3$, $\omega = 5$, and $\epsilon=0.01$. The numerical fit (dashed black) is also shown.}
\label{fig:hdotApproxSmallEps}
\end{figure}

\begin{figure}[htbp]
\vspace{3mm}
\centering
\includegraphics[width = 0.47\textwidth]{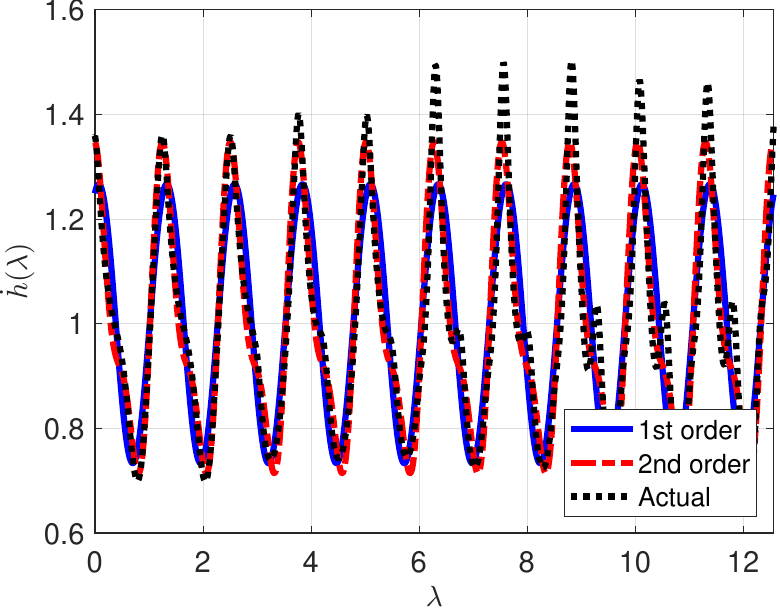}
\caption{First-order (solid blue) and second-order (dashed red) approximations of $\dot{h}(\lambda)$ with $\tau^* = \tau_0 = 3$, $\omega = 5$, and $\epsilon=0.1$. The numerical fit (dashed black) is also shown.}
\label{fig:hdotApproxLargeEps}
\end{figure}


We can assess the approximation error using the conditions
\eqref{eq:hcond1}-\eqref{eq:hcond2} required for a valid
time-transformation. Both the first- and second-order approximations for $h$ satisfy \eqref{eq:hcond2}. 
Thus, we only need to verify that $h$ satisfies \eqref{eq:hcond1} and that $h$ is strictly increasing.
Consider the case $\tau^* = \tau_0 = 3$ and $\omega = 5$. Then, $\dot{h}(\lambda) > 0$ for all $\lambda \geq -\tau^*$ holds for both the first- and second-order approximations if $\epsilon \in (0, 0.2)$. 
Next, we check the Abel equation \eqref{eq:hcond1}. If the approximation satisfies $\dot{h}(0) = \frac{\dot{h}(-\tau^*)}{1 - \dot{\tau}(0)}$ then we can use the seed propagation conditions \eqref{eq:condition on phi} and \eqref{eq:construct} to verify that $h$ is an actual time-transformation. Thus, we quantify the approximation error by $e := \big|\dot{h}(0)(1 - \dot{\tau}(0)) - \dot{h}(-\tau^*)\big|$, which measures the difference between the approximation and an actual time-transformation. Figure~\ref{fig:e vs eps} shows the dependence of the error $e$ on $\epsilon \in (0, 0.2)$. The error $e$ goes to zero quadratically and cubically for the first- and second-order approximations, respectively, as $\epsilon \to 0$. This behavior is consistent with the theoretical error bounds for the proposed expansions.



\begin{figure}[htbp]
\centering
\includegraphics[width = 0.47\textwidth]{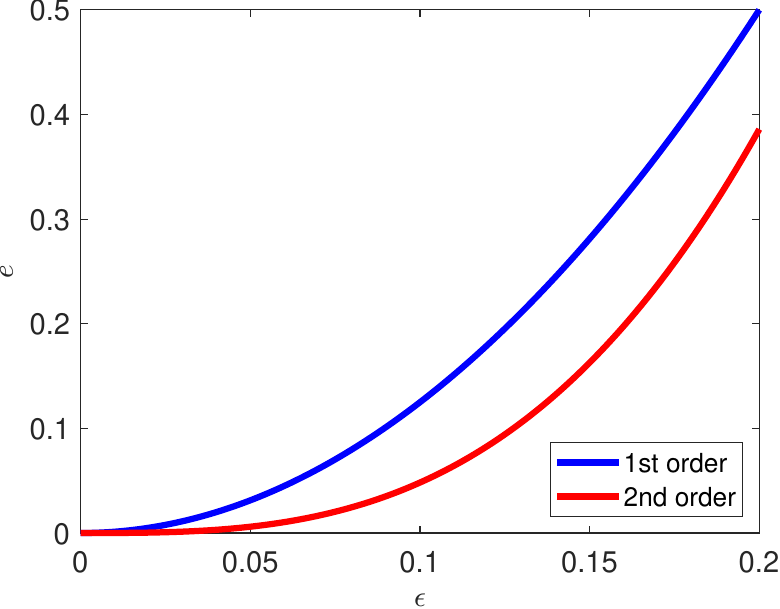}
\caption{Error $e$, the difference between the approximation and an actual time transformation, versus $\epsilon$.}
\label{fig:e vs eps}
\end{figure}

\subsection{Numerical Example}
\label{sec:stabnumex}

We consider the DDE \eqref{eq:DDEwithTVdelay} from \cite{gu2003stability} with:
\begin{align*}
A_0 &= \bmtx -2 & 0 \\ 0 & -0.9 \emtx
\mbox{ and } 
A_1 = \bmtx -1 & 0 \\ -1 & -1 \emtx.
\end{align*}
We again consider the sinusoidal delay given in 
Equation \eqref{eq:tauex} with the parameters 
$\tau_0 = 3$, $\omega = 5$, and $\epsilon \in (0, 0.2)$.
A second-order, time-transformation is constructed with
$\tau^* = \tau_0 = 3$. We use the analysis pipeline described in Section~\ref{sc:pipeline} with this time-transformation.  The DDE \eqref{eq:DDEwithTVdelay} is verified to be stable for $\epsilon = 0.01$, but the stability test fails for $\epsilon = 0.1$.

\section{Conclusion and Future work}
\label{sc:concfw}
This paper derived an explicit, although approximate, time-transformation for systems with a delay consisting of a constant plus small periodic term. We use a perturbative expansion to construct first- and second-order approximations for the time-transformation.
We provided a simple numerical example to illustrate the approach. We also demonstrate the use of this time-transformation to analyze stability of the system with this class of periodic delays. Future work will explore application of these approximations to stability and performance analysis.



\section*{Acknowledgment}
This material is based upon work supported by the National Science Foundation under Collaborative Grants No. 2337751 and 2337752. Any opinions, findings, and conclusions or recommendations expressed in this material are those of the author(s) and do not necessarily reflect the views of the National Science Foundation. 

\bibliographystyle{IEEEtran}
\bibliography{References}

\end{document}